\documentclass[11pt]{amsart}

\usepackage{amssymb,amsmath,amsthm}
\usepackage[hidelinks=true]{hyperref}

\theoremstyle{definition}
\newtheorem{dfn}{Definition}[section]
\newtheorem{thm}{Theorem}[section]
\newtheorem{pro}{Proposition}[section]
\newtheorem{exa}{Example}[section]
\newtheorem{cla}{Claim}[section]

\newcommand{\eg}{\textit{e.g.}}
\newcommand{\ie}{\textit{i.e.}}

\newcommand*{\markdef}[1]{{\normalfont\textbf{#1}}}
\newcommand*{\seq}[1]{\langle #1 \rangle}
\newcommand*{\R}{\mathbb{R}}
\newcommand*{\Z}{\mathbb{Z}}
\newcommand*{\N}{\mathbb{N}}
\newcommand*{\C}{\mathbb{C}}
\newcommand*{\E}{\mathbb{E}}

\DeclareMathOperator{\id}{id}

\begin{document}

\title{\bf Iterated Function Systems with Economic Applications}
\author{Shilei Wang}
\address{Universit\`{a} Ca' Foscari di Venezia, Venice, Italy}
\email{shilei.wang@unive.it}
\date{September, 2012}

\begin{abstract}
This work's purpose is to understand the dynamics of some social systems whose properties can be captured by certain iterated function systems. To achieve this intension, we start from the theory of iterated function systems, and then we study two specific economic models on random utility function and optimal stochastic growth.
\end{abstract}

\maketitle

\section{Iterated Function Systems}

\subsection{Discrete Dynamical Systems}

The transition rule in a discrete dynamical system can be determined solely by a single function mapping the state space into itself. Let the state space $X$ be a metric space with the metric $d:X\times X\to\R_+$, and let the time domain be $T=\Z$.

\begin{dfn}
A \markdef{discrete dynamical system} on $X$ is a pair $(X,f)$, where $f:X\to X$ is a continuous function, such that $x_{n+1}=f(x_n)$ for all $x_n,x_{n+1}\in X$ and all $n\in\Z$.
\end{dfn}

If we consider a state $x$, by the iterated transition rule, we can have a forward series of states as $f(x)$, $f(f(x))$, $f(f(f(x)))$, and so on, and also a backward series of states as $f^{-1}(x)$, $f^{-1}(f^{-1}(x))$, $f^{-1}(f^{-1}(f^{-1}(x)))$, and so on. To make our notations concise enough, we define $f^n$ in a recursive way, say $f^n=f^{n-1}\circ f$ for $n\in\Z$, and $f^0=\id_X$, where $\id_X$ is the identity function on $X$. We have a discrete motion passing through $x$ being
\[
\dotsc,f^{-n}(x),\dotsc,f^{-1}(x),x,f(x),\dotsc,f^n(x),\dotsc,
\]
where the transition rule between the time $m$ with a state $x\in X$ and the time $m+n$ is $f^n(x)$, for any $m,n\in\Z$.

The trajectory passing through $x$ is
\[
\gamma(x)=\bigcup_{n\in\Z}\{f^n(x)\},
\]
and similarly the positive and negative semi-trajectories are
\[
\gamma^+(x)=\bigcup_{n\in\Z_+}\{f^n(x)\},\ \text{and}\ \gamma^-(x)=\bigcup_{n\in\Z_-}\{f^n(x)\}.
\]
We know that $\gamma^+(x)\cup\gamma^-(x)=\gamma(x)$, and $\gamma^+(x)\cap\gamma^-(x)\supseteq\{x\}$. If $\gamma(x)=\{x\}$ or equivalently $f(x)=x$, we call $x$ is an equilibrium state, or in terms of functions, a fixed state.

A set of states $S\subseteq X$ is said to be invariant if $f(S)=S$, where $f(S):=\{f(x):x\in S\}$. The $\omega$-limit set of a state $x$ is $\omega(x)$, which contains all the $\omega$-limit states of $x$. If $\omega(x)\neq\emptyset$, then it must be invariant, say $f\big(\omega(x)\big)=\omega(x)$.

A state $x$ is said to be periodic, if there exists a $k\in\Z_+$ such that $f^k(x)=x$. The period $p$ of $\gamma^+(x)$ for any periodic state $x$ is defined as the minimal integer $k$ such that $f^k(x)=x$. 

If the period of $\gamma^+(x)$ is $1$, then $f(x)=x$, so $x$ is a fixed state. If the period of $\gamma^+(x)$ is $p<+\infty$, then $\gamma^+(x)=\{x,f(x),\dotsc,f^p(x)\}$. A state $x$ is said to be finally periodic, if there exists an $m\in\Z_+$ such that $f^n(x)$ is a periodic state for all $n\geq m$, or equivalently $f^{n+p}(x)=f^n(x)$ for all $n\geq m$, where $p$ is the period. $x$ is said to be asymptotically periodic, if there exists a $y\in X$ such that $\lim_{n\to+\infty}d\big(f^n(x),f^n(y)\big)=0$.

A set of states $A\subseteq X$ is called an attractor, if there exists a neighborhood $N(A,\varepsilon)$ of $A$ such that $f(N(A,\varepsilon))\subseteq N(A,\varepsilon)$, and
\[
\omega\big(N(A,\varepsilon)\big)=\bigcap_{n\in\Z_+}f^n\big(N(A,\varepsilon)\big)=A.
\]

Now consider a compact state space $X\subseteq\R$, we can state the following theorem that any periodic patterns can be implied by the existence of a period $3$ trajectory.

\begin{thm}
If $f:X\to X$ is continuous where $X$ is an interval in $\R$, and there is a trajectory of period $3$, namely there are three states $x,y,z\in X$ such that $f(x)=y$, $f(y)=z$, and $f(z)=x$, then there exist trajectories of period $n$, for any $n\in\N$.
\end{thm}

\begin{proof}
See the proof of Theorem $10.1$ in Devaney \cite{devaney86}.
\end{proof}

In general, this result is a special case of \v{S}arkovskii's theorem. Define the \v{S}arkovskii's ordering of all the natural numbers by $>_S$, if there is a trajectory of period $n$, then there also exist trajectories of period $m$, where $n>_S m$. Since $3$ is the first natural number in \v{S}arkovskii's ordering, and all the powers of $2$ in the order
\[
\dotsb >_S 2^n>_S 2^{n-1}>_S \dotsb >_S 2^2 >_S 2>_S 1
\]
are the last group of numbers in \v{S}arkovskii's ordering, a period $3$ trajectory implies all the other possible periods, and if we have a period not a power of $2$, then we should have all the periods of the powers of $2$.

The discrete dynamical system $(X,f)$ is chaotic, if
\begin{enumerate}
\item for any $x\in X$ and $\varepsilon>0$, there exist $\delta>0$ and $n\in\Z_+$ such that $d\big(f^n(x),f^n(y)\big)>\varepsilon$ for all $y\in N(x,\delta)$,
\item for any pair of open sets of states $S_1$ and $S_2$, there exists an $n\in\Z_+$ such that $f^n(S_1)\cap S_2\neq\emptyset$.
\end{enumerate}
So a chaotic discrete dynamical system means that its dynamics is sensitively dependent on the initial states, and its states are transitive.

We also want to give a definition of chaos in the sense of Li and Yorke \cite{li75}, in which $X\subseteq\R$ is compact, and $f:X\to X$ is of class $C^0$.

\begin{dfn}
A discrete dynamical system $(X,f)$ is said to be \markdef{non-periodically chaotic}, if there exists an uncountable set $S\subseteq X$ such that
\begin{enumerate}
\item $\limsup_{n\to+\infty}d\big(f^n(x),f^n(y)\big)>0$, and $\liminf_{n\to+\infty}d\big(f^n(x),f^n(y)\big)=0$ for all $x,y\in S$ with $x\neq y$,
\item for all $x,z\in S$ with $z$ is periodic, $\limsup_{n\to+\infty}d\big(f^n(x),f^n(z)\big)>0$.
\end{enumerate}
\end{dfn}

This definition of chaos is slightly weaker than the definition we give before. If a discrete dynamical system on $X\subseteq\R$ is chaotic, then it must be also non-periodically chaotic. On the other hand, if a discrete dynamical system is non-periodically chaotic, it is not necessary to be chaotic. Notice that the condition (ii) in the definition of non-periodic chaos means that there is no asymptotically periodic state in $X$. Suppose $u$ is a state not asymptotically periodic, so $\omega(u)$ has infinitely many states, and there exists a unique minimal invariant set $V\subseteq\omega(u)$, such that $V=\omega(v)$ having infinitely many states, where $v\in X$. Consider the set of states $U=X\setminus V$, then $f^n(V)\cap U=\emptyset$ for all $n\in\Z_+$, so it is not transitive, which hence means that it is not chaotic.

\subsection{Iterated Function Systems}\label{sec3.2}

In this section, we consider a set of contractive functions defined on the state space $X$ with a metric $d:X\times X\to\R_+$, in which a function $f:X\to X$ is said to be contractive if there exists a $\lambda\in(0,1)$ such that $d\big(f(x),f(y)\big)\leq\lambda d(x,y)$ for all $x,y\in X$. In the following paragraphs, we denote $I_N$ as an index set with $N$ elements, where $N\geq 2$ is finite.

\begin{dfn}
Let $F$ be a collection of $N\geq 2$ continuous contractive functions, say $F=\{f_i:i\in I_N\}$ where $f_i:X\to X$. Then the pair $(X,F)$ is said to be an \markdef{iterated function system}, if $(X,f_i)$ is a discrete dynamical system for each $i\in I_N$.
\end{dfn}

Suppose $(X,d)$ is compact, and $\mathcal{Q}(X)$ is the set of all nonempty compact subsets of $X$. For any $U,V\in\mathcal{Q}(X)$, the Hausdorff metric between $U$ and $V$ is
\[
d_H(U,V)=\sup_{u\in U,v\in V}\{d(u,V),d(v,U)\},
\]
where $d(x,Y):=\inf_{y\in Y}d(x,y)$ is the metric between a state $x\in X$ and a nonempty set $Y\subseteq X$. Then $\mathcal{Q}(X)$ with the Hausdorff metric $d_H:\mathcal{Q}(X)\times\mathcal{Q}(X)\to\R_+$ is a compact metric space.

Consider a mapping $H:\mathcal{Q}(X)\to\mathcal{Q}(X)$ such that for any $B\in\mathcal{Q}(X)$, we have \[
H(B)=\bigcup_{i\in I_N}f_i(B),
\]
where $H$ is called a Hutchinson operator. Define $H^n$ in a recursive way, say $H^n=H^{n-1}\circ H$ for $n\in\Z$, and $H^0=\id_{\mathcal{Q}(X)}$, where again $\id_{\mathcal{Q}(X)}$ is an identity mapping on $\mathcal{Q}(X)$.

\begin{dfn}
$A\in\mathcal{Q}(X)$ is called an \markdef{attractor} of the iterated function system $(X,F)$, if $H(A)=A$, and there exists a neighborhood $N(A,\varepsilon)\in\mathcal{Q}(X)$ such that $H\big(N(A,\varepsilon)\big)\subseteq N(A,\varepsilon)$ and $\bigcap_{n\in\Z_+}H^n\big(N(A,\varepsilon)\big)=A$.
\end{dfn}

Obviously, if $A$ is an attractor of $(X,F)$, then there exists a neighborhood $N(A,\varepsilon)\in\mathcal{Q}(X)$ such that $\lim_{n\to+\infty}H^n\big(N(A,\varepsilon)\big)=A$. The largest such neighborhood is called the basin of the attractor $A$, denoted as $B(A)$.

\begin{pro}\label{pro3.1}
The mapping $H:\mathcal{Q}(X)\to\mathcal{Q}(X)$ has a unique fixed point $A\in\mathcal{Q}(X)$, \ie\ $H(A)=A$, and hence the iterated function system $(X,F)$ has a unique attractor $A$.
\end{pro}

\begin{proof}
Since $f_i$ is a contractive function, there exists a $\lambda_i$ such that $d\big(f_i(x),f_i(y)\big)\leq\lambda_i d(x,y)$ for all $x,y\in X$, and any $i\in I_N$. Define $\lambda=\max_{i\in I_N}\lambda_i$, then 
\[
d_H\big(H(U),H(V)\big)\leq\sup_{i\in I_N}d_H\big(f_i(U),f_i(V)\big)\leq\sup_{i\in I_N}\big\{\lambda_i d_H(U,V)\big\}\leq\lambda d_H(U,V),
\]
for all $U,V\in\mathcal{Q}(X)$. So $H$ has a unique fixed point, say $A$, which is also the attractor of $(X,F)$.
\end{proof}

Now consider the space $I_N^\infty:=I_N\times I_N\times\dotsb$, and for any $\mu\in I_N^\infty$, we write $\mu=\seq{\mu_1,\mu_2,\dotsb,\mu_k,\dotsb}$, where $\mu_k\in I_N$ for all $k\in\N$. For any $\mu,\upsilon\in I_N^\infty$, define the Baire metric as $d_B(\mu,\upsilon)=2^{-j}$ where $j$ is the first index when $\mu_k\neq\upsilon_k$. Now we have a compact metric space $(I_N^\infty,d_B)$. For any given set of states $S\subseteq X$, we can define a mapping $C:I_N^\infty\to\mathcal{Q}(X)$ such that
\[
C(\mu,S)=\bigcap_{k\in\N}f_{\mu_k}\circ f_{\mu_{k+1}}\circ\dotsb\circ f_{\mu_\infty}(S),
\]
where $\mu$ defines a trajectory starting from any state $x\in S$ as
\[
\gamma^+(x)=\{f_{\mu_1}(x)\}\cup\{f_{\mu_1}\circ f_{\mu_2}(x)\}\cup\dotsb\cup\{f_{\mu_1}\circ f_{\mu_2}\circ\dotsb\circ f_{\mu_k}(x)\}\cup\dotsb.
\]

If $S\subseteq B(A)$, then for any $\mu\in I_N^\infty$, we have $C(\mu,S)\subseteq A$, and especially we have $C(I_N^\infty,B(A))=A$. So we claim that the attractor of an iterated function system can be achieved by all the possible trajectories determined by $\mu\in I_N^\infty$, in which the transition function is $f_{\mu_k}$ at time $k$, where $\mu_k\in I_N$. 

We can also impose a probability measure on $I_N^\infty$ such that $\mu$ can be realized with some probability density. In a simplified situation, the probability measure over the whole space $I_N^\infty$ could be represented by a stationary discrete probability measure $\pi$ over $I_N$ such that $\sum_{i\in I_N}\pi(i)=1$, where $\pi$ is independent with the time and the state. At any time $k\in\Z$, the transition rule is chosen at random from $F$ according to the probability measure $\pi$ over $I_N$.

\begin{dfn}
Let $(X,F)$ be an iterated function system, where $F=\{f_i:i\in I_N\}$ where $f_i:X\to X$. If $\pi$ is probability measure on $I_N$ such that $\sum_{i\in I_N}\pi(i)=1$, then the triplet $(X,F,\pi)$ is said to be an \markdef{iterated random function system}.
\end{dfn}

We use the random variable $\sigma_k$ to represent the index chosen at time $k$, and the related random transition rule is $f_{\sigma_k}$ at time $k$, so we have the random iteration
\[
Z_{k+1}=f_{\sigma_{k+1}}(Z_k),
\]
where $Z_{k}$ is the random state with the support of $X$ for all $k\in\N$. Being similar with the above notations, $\sigma$ is defined as $\seq{\sigma_1,\sigma_2,\dotsc,\sigma_k,\dotsc}$.

If the initial state is $x$, its random trajectory is written as $\Gamma^+(x)=\bigcup_{k\in\N}\{Z_k\}$, where $Z_1=f_{\sigma_1}(x)$, and $Z_k=f_{\sigma_k}(Z_{k-1})$ for $k\in\N\setminus\{1\}$. Clearly, the stochastic process $(Z_k)_{k\in\N}$ is a Markov chain. Given a state $x\in X$ at time $k$, the next state located in $Y\subseteq X$ has a probability
\[
P(x,Y)=\sum_{i\in I_N}\pi(i)\mathbf{1}_Y(f_i(x)),
\]
where
\[\mathbf{1}_Y(y)=\left\{
\begin{aligned}
1,\quad & \text{if}\ y\in Y\\
0,\quad & \text{if}\ y\notin Y
\end{aligned}
\right..\]

For any Borel subsets $Y$ of $X$, there exists an invariant probability measure $\rho$ such that
\[
\rho(Y)=\int_X P(x,Y)d\rho(x)=\sum_{i\in I_N}\pi(i)\rho(f_i^{-1}(Y)).
\]
Such a probability measure is called the $\pi$-balanced measure on $(X,F,\pi)$, as discussed in Barnsley and Demko \cite{barnsley85}. The support of $\rho$, defined as $R(\rho)=\{x\in X:\rho(x)\neq 0\}$, satisfies $R(\rho)=\bigcup_{i\in I_N} f_i\big(R(\rho)\big)$, and hence $R(\rho)=A$ by Proposition \ref{pro3.1}, which means that the support of a $\pi$-balanced measure on $(X,F,\pi)$ is the unique attractor $A$ of $(X,F)$ for any possible $\pi$.

Therefore, we claim that the attractor in $(X,F)$ can be achieved by a random iteration generated from $(X,F,\pi)$ after a sufficiently long time. For any given state $x\in X$, $\bigcap_{k\in\N}\Gamma^+(Z_k)=A$ with probability $1$, where $Z_k\in\Gamma^+(x)$ for all $k\in\N$, which means that there exists a random trajectory starting from some states in another random trajectory whose initial state is given at the beginning, and this random trajectory will cover the attractor eventually. Or in terms of limit states, the random $\omega$-limit set of any given initial state $x$, say $\Omega(x)$, will be equal to $A$ for sure.

Now consider a special case in which the state space $X$ is a subset of the $n$-dimensional Euclidean space $\E^n$, and $f_i:X\to X$ is affine, say $f_i(x)=D_ix+e_i$, where $x\in X$ is an $n\times 1$ vector, $D_i$ in an $n\times n$ matrix, and $e_i$ is an $n\times 1$ vector, where $i\in I_N$. Suppose again there is a probability measure $\pi$ over $I_N$, then the random iteration is $Z_{k+1}=DZ_k+e$, where $D$ and $e$ depends on the randomly chosen index $\sigma_k\in I_N$ according to $\pi$. We know that $(Z_k)_{k\in\N}$ is a Markov chain, so this random iteration can be equivalently represented as a general stochastic process $Z_{k+1}=D_{k+1}Z_k+e_{k+1}$, where $Z_k$ is the random state with the support of $X$, and $D_k,e_k$ are i.i.d. random matrices.

\begin{exa}\label{exa3.1}
Suppose $X=[0,1]$, $f_a(x)=x/3$, $f_b(x)=x/3+2/3$, $F=\{f_a,f_b\}$, $I_2=\{a,b\}$, and $\pi(a)=\pi(b)=0.5$. This iterated random function system is equivalent with the following autoregressive process
\[
Z_{k+1}=Z_k/3+\varepsilon_{k+1},\ k\in\Z_+,
\]
where $Z_0=x\in[0,1]$ is the given initial state, and $\varepsilon_k$ is an i.i.d. random variable taking $0$ and $2/3$ with equal probability. $(X,F)$ has an attractor
\[
A=\Big\{\sum_{i=1}^\infty x_i/3^i:x_i\in\{0,2\},\ \forall i\in\N\Big\},
\]
which is actually the Cantor ternary set. 

Let $B_n:=\big\{\sum_{i=n}^\infty x_i/3^i:x_i\in\{0,2\},\ \forall i\in\N\big\}$, then $A=B_1$, and $f_a(A)=A/3=B_2=B_2+0/3$, $f_b(A)=A/3+2/3=B_2+2/3$, where $A/3=\{a/3:a\in A\}$, and $B_2+c:=\{b+c:b\in B_2\}$ for $b=0,2/3$, so $f_a(A)\cup f_b(A)=B_1=A$, which confirms that $A$ is the unique attractor of $(X,F)$.

The attractor $A$ is the support of a $\pi$-balanced measure on $(X,F,\pi)$, and thus it can be achieved by a random trajectory $\Gamma^+(x)$ starting from any $x\in[0,1]$. The random trajectory can be represented by the associated autoregressive process. 

If the initial state $x\in A$, then we have $\Gamma^+(x)=A$ for sure. If $x\in[0,1]\setminus A$, then there exists $x_1,x_2,\dotsc,x_m\in\{0,2\}$ such that $x=\sum_{i=1}^m x_i/3^i+m(x)$, where $m(x)\leq 1/3^m$. If $m$ is sufficiently large, then $m(x)$ is close enough to $0$, while if $m\to+\infty$, $m(x)\to 0$. $Z_1=x+\varepsilon_1=y+m(x)/3$, where $y=\sum_{i=1}^m x_i/3^i+\varepsilon_1\in A$, and generally, $Z_{k}=z+m(x)/3^k$, where again $z\in A$, so after a finite period of time, say $l$, we will have $Z_l\in A$, and hence we have $\Gamma^+(Z_l)=A$ with probability $1$. Thus for any $x\in[0,1]$, this stochastic process will cover $A$ at last, say $\Omega(x)=\bigcap_{k\in\N}\Gamma^+(Z_k)=A$, where $Z_k\in\Gamma^+(x)$.
\qed
\end{exa}

\subsection{Strange Attractors}

The concept of strange attractors first appeared in a paper on turbulence by Ruelle and Takens \cite{ruelle71}, but was not defined precisely. A definition of strange attractors later was given in a popular article by Ruelle \cite{ruelle80}, but according to that definition, the attractors of any chaotic dynamical systems are strange. We want to restrict this definition by stating that an attractor is strange if the dynamical system is chaotic, and it should a fractal. The fractals are defined by Mandelbrot as

\begin{dfn}
A set is said to be a \markdef{fractal}, if its Hausdorff-Besicovitch dimension strictly exceeds its topological dimension.
\end{dfn}

Let $(X,d)$ be a metric space as usual, and $Z\subseteq X$, the topological dimension of $Z$ denoted as $\dim_L(Z)$ is an integer belongs to $\Z_+\cup\{-1,+\infty\}$. If $Z=\emptyset$, then $\dim_L(Z)=-1$. If $Z$ is a discrete set, then $\dim_L(Z)=0$. In general, $\dim_L(Z)=n$, if for any open cover of $Z$, any state $z\in Z$ can be contained in $n+1$ open sets of that cover at most. If there does not exist such a finite $n$, then $\dim_L(Z)=+\infty$.

For any $d\geq 0$, we can define the $\delta$-approximate $d$-dimensional Hausdorff measure of $Z$ as
\[
\mathcal{H}_\delta^d(Z)=\inf\Big\{\sum_{i\in I}\big(\text{diam}(U_i)\big)^d:\bigcup_{i\in I}U_i\supseteq Z,\ \text{diam}(U_i)\leq\delta\Big\},
\]
where $\text{diam}(U_i):=\sup\{d(z_1,z_2):z_1,z_2\in U_i\}$, $I$ is a countable index set, and $\{U_i\}_{i\in I}$ is an open $\delta$-cover of $Z$. The $d$-dimensional Hausdorff measure of $Z$ is defined as $\mathcal{H}^d(Z)=\lim_{\delta\to 0}\mathcal{H}_\delta^d(Z)$. Then the Hausdorff-Besicovitch dimension of $Z$ is defined as $\dim_H(Z)$ such that
\[
\dim_H(Z)=\inf\{d\in\R_+:\mathcal{H}^d(Z)=0\}=\sup\{d\in\R_+:\mathcal{H}^d(Z)=\infty\}.
\]

In practice, when we need to determine the dimension of a fractal, the more useful notion is the box-counting dimension. The upper and lower box-counting dimensions of $Z\subseteq X$ are defined as
\[
\overline{\dim}_B(Z)=\limsup_{\varepsilon\to 0}\frac{\log N(Z,\varepsilon)}{\log(1/\varepsilon)},\ \text{and}\ \underline{\dim}_B(Z)=\liminf_{\varepsilon\to 0}\frac{\log N(Z,\varepsilon)}{\log(1/\varepsilon)},
\]
where $N(Z,\varepsilon)$ in the minimal number of sets as required for an open $\varepsilon$-cover of $Z$. As we suppose $N(Z,\varepsilon)\sim \lambda (1/\varepsilon)^d$ by extending the observations of regular cases in the Euclidean space, we have $\log N(Z,\varepsilon)=\log\lambda+d\log(1/\varepsilon)$, where $\lambda$ is a positive constant. That's why we define the box-counting dimension by the limit of the ration between $\log N(Z,\varepsilon)$ and $\log(1/\varepsilon)$.

Once we have $\overline{\dim}_B(Z)=\underline{\dim}_B(Z)$ for a set of states $Z$, we say there exists the box-counting dimension, which is denoted as
\[
\dim_B(Z)=\overline{\dim}_B(Z)=\underline{\dim}_B(Z).
\]

In general, we have $\dim_H(Z)\leq\underline{\dim}_B(Z)\leq\overline{\dim}_B(Z)$, and $\dim_H(Z)\leq\dim_B(Z)$ if $\underline{\dim}_B(Z)=\overline{\dim}_B(Z)$. For the unique attractor $A$ of an iterated function system $(X,F)$, if $X\subseteq\E^n$ (or in some other special settings), we have $\dim_H(A)=\dim_B(A)$, and hence we can calculate the box-counting dimension of $A$ to get exactly its Hausdorff-Besicovitch dimension. If $A$ is strange, we have $\dim_H(A)=\dim_B(A)\geq\dim_L(A)$.

Consider an iterated function system $(X,F)$ has a set of affine functions with similarities $f_i:\R^n\to\R^n$ for all $i\in I_N$. We know $\bigcup_{i\in I_N}f_i(A)=A$, where $A$ is the attractor of $(X,F)$. If we assume $f_i(A)$'s are not overlapping, which means that there exists an open subset $D$ of $A$ such that $f_i(D)\cap f_j(D)=\emptyset$ for all $i,j\in I_N$ with $i\neq j$, and $\bigcup_{i\in I_N}f_i(D)\subseteq D$, then the similarity dimension of $A$, which is equal to $\dim_B(A)$ and $\dim_H(A)$, is the unique solution of the equation
\[
\sum_{i\in I_N}\lambda_i^d=1,
\]
where $\lambda_i<1$ is the Lipschitz constant of the function $f_i$ for all $i\in I_N$. 

The intuitive proof of the above equation is clear. Suppose $N(A,\varepsilon)$ is again the minimal number of sets required to fully $\varepsilon$-cover $A$, then $N(A,\varepsilon)\sim\lambda(1/\varepsilon)^d$. Notice that $A=\bigcup_{i\in I_N}f_i(A)$, and $f_i(A)$'s are not overlapping, so we have $N(A,\varepsilon)=\sum_{i\in I_N}N\big(f_i(A),\varepsilon\big)$. But $N\big(f_i(A),\varepsilon\big)=N(A,\varepsilon/\lambda_i)$ as $A$ is scaled by the Lipschitz constant $\lambda_i$ of $f_i(A)$. Then $N(A,\varepsilon)=\sum_{i\in I_N}N(A,\varepsilon/\lambda_i)$, and hence $(1/\varepsilon)^d=\sum_{i\in I_N}(\lambda_i/\varepsilon)^d$ when $\varepsilon\to 0$, so we get $\sum_{i\in I_N}\lambda_i^d=1$.

In the following paragraphs, we will consider some examples. At first we want to revisit the Cantor set that has been discussed in Example \ref{exa3.1}.

\begin{dfn}
Let $D_i=\{0,1\}$ for all $i\in\N$ with the discrete topology, and write $\{0,1\}^\omega=\prod_{i\in\N}D_i$, then any set homomorphic to $\{0,1\}^\omega$ is called a \markdef{Cantor set}.
\end{dfn}

\begin{exa}
Let $X=[0,1]$, and define a real number $c\in(2,+\infty)$. We remove the interval $(1/c,1-1/c)$ from $[0,1]$ having
\[
C_1=[0,1/c]\cup[1-1/c,1],
\]
and we delete the $(1-2/c)$ ratio of middle parts in $[0,1/c]$ and $[1-1/c,1]$ respectively obtaining
\[
C_2=[0,1/c^2]\cup[1/c-1/c^2,1/c]\cup[1-1/c,1-1/c+1/c^2]\cup[1-1/c^2,1].
\]
In general, we have $C_{k+1}=C_k/c\cup \big(C_k/c+(1-1/c)\big)$. When $k\to+\infty$, we have the Cantor set
\[
C=\Big\{\sum_{i\in\N}x_i/c^i:x_i\in\{0,c-1\},\ \forall i\in\N\Big\},
\]
which is the attractor of $(X,F)$ with $f_a(x)=x/c$ and $f_b(x)=x/c+(1-1/c)$.

If taking $\varepsilon=1/c^k$, we have $N(C,\varepsilon)=2^k$ as we need $2^k$ open sets to $\varepsilon$-cover $C$ with $2^k$ disconnected intervals. By the definition of the box-counting dimension, we have $\dim_H(C)=\dim_B(C)=\lim_{k\to\infty}\log 2^k/\log c^k=\log 2/\log c$.

Alternatively, since $C$ is the attractor of an affine function systems with similarities, and we have their scaling numbers $\lambda_a=\lambda_b=1/c$, so $\dim_B(C)$ is the unique solution of the equation $(1/c)^d+(1/c)^d=1$. Thus $\dim_B(C)=\log_c 2=\log2/\log c$. Since $c\in(2,+\infty)$, $\dim_B(C)\in(0,1)$. In the case of Cantor ternary set in Example \ref{exa3.1}, where $c=3$, its dimension is $\dim_H(A)=\dim_B(A)=\log 2/\log 3$.
\qed
\end{exa}

\begin{exa}
Let $X\subseteq\C$ is a triangle with vertices $0$, $a$, and $1$, where $a$ is a complex number. Consider an iterated function system $(X,F)$, where $F=\{f_1,f_2\}$ with $f_1(z)=a\overline{z}$, and $f_2(z)=(1-a)\overline{z}+a$, where $a$ satisfies $\vert a\vert<1$ and $\vert 1-a\vert <1$, and $\overline{z}$ is the complex conjugate of $z$. (This case was proposed by De Rham \cite{derham57}.)

If $a=1/2+i\sqrt{3}/6$, then $X$ is equilateral, and the attractor of this iterated function system is the Koch curve $K$, which can be constructed by replacing the middle $1/3$ parts of each line segment of $X$ with a new equilateral triangle, and doing this process with iteration. In this case, $\vert a\vert=\vert 1-a\vert=1/\sqrt{3}$, so their scaling numbers are $\lambda_1=\lambda_2=1/\sqrt{3}$, and its similarity dimension is the unique solution of the equation $2(1/\sqrt{3})^d=1$. Thus $\dim_H(K)=\dim_B(K)=\log4/\log 3$.

If $\vert a-1/2\vert=1/2$, \eg\ $a=1/2+i/2$, then the attractor of $(X,F)$ is the Peano curve $P$, which fully fill the space $X$. We have $\lambda_1=\lambda_2=1/\sqrt{2}$, and solve the equation $2(1/\sqrt{2})^d=1$, we obtain $\dim_H(P)=\dim_B(P)=2$.
\qed
\end{exa}

\section{Random Utility Functions}

We want to consider an agent $w$ in a large population $W$, who has a rational preference relation $\succsim$, but can be also affected by some psychological factors when making her choices. Concretely, we assume the decision state space is $X$, and the preference relation is a weak order $\succsim$ on $X$ such that
\begin{enumerate}
\item either $x\succsim y$ or $y\succsim x$ for all $x,y\in X$,
\item $x\succsim y$ and $y\succsim z$ implies $x\succsim z$ for all $x,y,z\in X$.
\end{enumerate}

If $x\succsim y$ and not $y\succsim x$, we say $x\succ y$. Then we have the following preference representation result.

\begin{cla}
If $X$ is countable, or $X$ is uncountably infinite but $\succsim$ is dense on $X$, which means that for any $Z\subset X$, when $x,y\in X\setminus Z$ such that $x\succ y$, there is a $z\in Z$ satisfying $x\succ z$ and $z\succ y$, then there exists a utility function $u:X\to\R$ such that $x\succsim y$ if and only if $u(x)\geq u(y)$ for all $x,y\in X$.
\end{cla}

Let $\mathcal{P}(X)$ be the power set of $X$. A function $C:\mathcal{P}(X)\to\mathcal{P}(X)$ is called a choice function, such that for any $Z\in\mathcal{P}(X)$ we have $C(Z)\neq\emptyset$ and $C(Z)\subseteq Z$. Now consider an arbitrary subset of $X$, say $Y\in\mathcal{P}(X)$, if $y\in C(Y)$, then we should have $u(y)\geq u(z)$ for all $z\in Y$, and if $u(y)<u(z)$, then $z\notin Y$.

When we observe the agent $w$ chooses $y$ from a set of states $Y$, we claim that $u(y)\geq u(z)$ for any $z\in Y$ is logically true. However, from the view of statistics, this inference is true only in a certain confidence level. So behind our observations, the utility of the agent $w$ and also her preference is random for us, in the sense that the true utility values are adapted by noises. Such a problem was first considered by Manski \cite{manski75}, in which he introduced the method of maximum score to estimate the stochastic utility model of choice from the observation data.

Define $v(x)$ is the observed utility when the agent chooses $x$, then we have $u(x)=v(x)+\varepsilon(x)$, where $\varepsilon(x)$ is the introduced noise related to $x$. If $v(x)$ is assumed to have a representation of linear form of $m$ observable variables, say $J(x)\in\R^m$, then $v(x)=J(x)'\beta$, where $\beta\in\R^m$ is the parameter to be estimated. We have
\[
u(x)=J(x)'\beta+\varepsilon(x),
\]
where the noises $\varepsilon(x)$ are i.i.d. across different alternatives $x$. The estimation $\hat{\beta}$ is a function of the observed data $J(x)$, where $x$ belongs to a given sample, say $S\subseteq X$.

Another way to deal with the difference between the utility associated with the observed choices and the true utilities is to think that the utility function of $w$ is random for her, but the random utility function should have a fundamental unchanged part at any time. So $w$'s knowledge of her own preference $\succsim$ is probabilistic, and she uses a random utility function to make choices. The assumption that all the possible utility functions have a common basic parts is set to capture the fact that she is not totally ambiguous with her own preference, but her knowledge is in a state between ignorance and certainty.

Formally, suppose that there is a finite index set $I_N$ with $N\geq 2$ elements, and the real-valued function $u_i:X\to\R$ is a possible utility function for any $i\in I_N$. Define a function $k:X\to\R$, we call it a fundamental utility function for the agent $w$, if there exists a contractive function $f_i:X\to X$ such that $u_i=k\circ f_i$ for all $i\in I_N$.

So we have $u_i(x)=f_i(k(x))$ for all $i\in I_N$. Denote the image of the fundamental utility function $k(x)$ on the domain $X$ as $k(X)\subseteq\R$. Then the random utility functions $\{u_i:i\in I_N\}$ on $X$ is equivalent with the function system $\{f_i:i\in I_N\}$ on $k(X)$. 

Now we introduce the time, and suppose the agent $w$ makes her choices along the time domain $\Z_+$. At time $t=0$, the fundamental utility function is $k_0(x)=k(x)$ for any $x\in X$. At time $t=1$, the fundamental utility function becomes $k_1(x)=f_i(k(x))$, which is also the utility function at time $t=0$. In general, at time $t=n$, the fundamental utility function is $k_n(x)=f_i(k_{n-1}(x))$ for any $i\in I_N$ and all $n\in\N$. This dynamics of fundamental utility functions can be thought of as a reasonable extension of the discounted utility function in the framework of a stationary time preference $(\succsim_t)_{t\in\Z_+}$, where $\succsim_t=\succsim$ for all $t\in\Z_+$ in our settings.

This construction can help to generate an iterated function system $(k(X),F)$, where $F=\{f_i:i\in I_N\}$ with $f_i:k(X)\to k(X)$, and $k:X\to\R$ is a given function. Since $u_i(x)$ at time $n$ is the fundamental utility function $k_{n+1}(x)$ at next time $n+1$, the iterated function system $(k(X),F)$ produces the values of both the fundamental utility and the actual utility. By the proposition \ref{pro3.1}, there exists an attractor $A\subseteq k(X)$ such that $A=\bigcup_{i\in I_N}f_i(A)$. Therefore, there exist some stable utilities under the iteration driven by $F$.

If there is a probability measure $\pi$ on $I_N$ such that $\sum_{i\in I_N}\pi(i)=1$, the iterated function system $(k(X),F)$ becomes an iterated random function system $(k(X),F,\pi)$, in which we have the following random iterated process
\[
U_{n}=f_{\sigma_{n}}(U_{n-1}),\ \text{for\ all}\ n\in\Z_+,
\]
where $U_n$ is the random utility at time $n$, $\sigma_n$ is the random index chosen from $I_N$ at time $n$ for all $n\in\Z_+$, and $U_{-1}=k(x)$. Denote $K_n$ as the random fundamental utility at time $n$, since $K_{n+1}=U_n$, we have $K_{n+1}=f_{\sigma_n}(K_n)$ for all $n\in\Z_+$ with $K_0=k(x)$.

Suppose $f_i(x)=\rho_i x$, where $0<\rho_i<1$ for all $i\in I_N$, and $\rho_i\neq\rho_j$ for any $i\neq j$, the iterated stochastic process is
\[
U_n=\xi_n U_{n-1},\ \text{for\ all}\ n\in\Z_+,
\]
where again $U_n$ is the random utility at time $n$, and $\xi_n$ is an i.i.d. random variable with the support $\{\rho_i:i\in I_N\}$ for all $n\in\Z_+$. For any state $x\in X$, at time $n\in\Z_+$, the random utility is
\[
U_n=\Big(\prod_{i=0}^n\xi_i\Big)k(x)=\exp\Big(\sum_{i=0}^n\log\xi_i\Big)k(x)=\exp\Big(-\sum_{i=0}^n\log(1/\xi_i)\Big)k(x),
\]
where $U_n\to 0$ almost surely, when $n\to+\infty$.

If $f_i(x)=\rho x+r_i$, where $0<\rho<1$ and $r_i>0$ for all $i\in I_N$, with $r_i\neq r_j$ for any $i\neq j$, the iterated stochastic process is
\[
U_n=\rho U_{n-1}+\varepsilon_n,\ \text{for\ all}\ n\in\Z_+,
\]
where $\varepsilon_n$ is again an i.i.d. random variable taking values from the set $\{r_i:i\in I_N\}$ for all $n\in\Z_+$. At time $n\in\Z_+$, the random utility is
\[
U_n=\rho^nk(x)+\sum_{i=0}^n\rho^{n-i}\varepsilon_i,
\]
in which $\rho^nk(x)$ vanishes when $n\to+\infty$, however the second part does not converge almost surely, simply as the new randomness $\varepsilon_n$ will enter into $U_n$ at each time $n$. Roughly speaking, in this iterated affine function system $(k(X),F)$, the attractor will be a fractal for most possible values of the parameters $\rho$ and $r_i$'s.

\section{Optimal Growth under Uncertainty}

In this section, we want to consider the nature of the steady state in a stochastic growth model. Consider an economy with a gross production function $Y=F(K,L)$, where $Y$ is the total production in the economy, $K$ is the capital input, and $L$ is the labor supply. Suppose this gross production function is a homogeneous function of degree $1$, then we have $Y/L=F(K/L,1)=f(K/L)$. Define $y=Y/L$ and $k=K/L$, we have the production function for a representative agent $w$ in that economy $y=f(k)$, which satisfies the classical hypothesis $f(0)=0$, and $f'(k)>0$, $f''(k)<0$, for all $k>0$, and the Inada conditions of $\lim_{k\to 0^+}f'(k)=+\infty$, and $\lim_{k\to+\infty}f'(k)=0$.

Suppose there is a stochastic factor $\xi$ can affect the production, we write the new production function as $y=f(k,\xi)$. In general, the shock $\xi$ can be either multiplicative or additive with $f(k)$, similar with the studies carried by Mitra, Montrucchio, and Privileggi \cite{mitra04}, we consider the case of multiplicative shocks. So we can write $y=f(k,\xi)=\xi f(k)$, where $\xi>0$ is a random shock. Suppose $\xi$ can take values in the set $\{\lambda_i:i\in I_N\}$, where $I_N$ is a finite index set with $N\geq 2$, with the probability measure $\pi$ on $I_N$ such that $\sum_{i\in I_N}\pi(i)=1$, and $\pi(i)>0$ for all $i\in I_N$.

The sustainability of this economy is possible by introducing the consumption $C$ and the investment $I$, both of whom have the same measurement with the production $Y$. Again we consider the consumption and the investment of the representative agent $w$, and hence we write her consumption as $c=C/L$, and her investment as $i=I/L$. 

After introducing the time domain $\Z_+$, we have the set of variables $\{y_n,k_n,i_n,c_n,\xi_n\}$ at each time $n\in\Z_+$. In this economy without the government, we have the following system
\[\left\{
\begin{aligned}
& y_n = c_n+i_n\\
& i_n = k_{n+1}\\
& y_{n} = \xi_{n}f(k_{n})
\end{aligned}
\right.,\]
where at time $t=0$, $k_0=i_{-1}$ is the initial capital, which produce $y_0=\xi_0f(k_0)$, and $y_0$ will be divided into the consumption $c_0$ and the investment $i_0=k_1$. In general, at time $t=n$, $k_n=i_{n-1}$ is the capital, and the production is $y_n=\xi_nf(k_n)$, which is the source of the consumption $c_n$ and the next time capital $k_{n+1}=i_n$, where $\xi_n$ is the random shock $\xi$ for any $n\in\Z_+$, so the process $(\xi_n)_n$ is i.i.d. on the time domain $\Z_+$.

Since $i_n$ and $k_{n+1}$ are related by $i_n=k_{n+1}$, and $c_n$ and $i_n$ are related by $c_n=y_n-i_n$ given $y_n$, we actually have the following three equivalent systems, which have the sets of variables $\{y_n,k_{n+1},\xi_n\}$, $\{y_n,c_n,\xi_n\}$, and $\{y_n,i_n,\xi_n\}$ respectively.

Suppose $w$ has a stationary utility function $u(c)$ with the properties that $u'(c)>0$, $u''(c)<0$, and $\lim_{c\to 0^+}u'(c)=+\infty$, which means $c_n>0$ at each time $n\in\Z_+$. $w$'s time preference is captured by a discounting factor $\rho\in(0,1)$, so her utility along $\Z_+$ can be written as $\sum_{n\in\Z_+}\rho^n u(c_n)$, if given her consumption flow $(c_0,c_1,\dotsc,c_n,\dotsc)$.

The steady states of this economic system are the solutions for the following dynamic optimization problem
\[
\begin{aligned}
\max\ & \mathbb{E}_0\sum_{n\in\Z_+}\rho^n u(c_n)\\
\textit{s.t.}\ & k_{n+1}=\xi_nf(k_n)-c_n\\
& k_0>0\ \text{is\ given}
\end{aligned}\quad ,
\]
where $\xi_n$ is an i.i.d. random variable with a finite support $\{\lambda_i:i\in I_N\}$ and a probability measure $\pi$, such that $\sum_{i\in I_N}\pi(i)=1$.

The optimal consumption flow of $w$ is captured by the following Euler equation
\[
u'(c_n)=\mathbb{E}_t\big(\rho\xi_{n+1}f'(k_{n+1})u'(c_{n+1})\big).
\]
Noticing that $k_{n+1}=y_n-c_n$, we can rewrite the above equation as
\[
u'(c_n)=\rho f'(y_n-c_n)\mathbb{E}_t\big(\xi_{n+1}u'(c_{n+1})\big).
\]
So for the optimal steady state along time, there exists a function $\varphi:\R\to\R$ such that $c_n=\varphi(y_n)$, and hence $k_{n+1}=y_n-\varphi(y_n)$. Therefore, the optimal state of the production in this economy has the following dynamics
\[
y_{n+1}=\xi_{n+1}f(k_{n+1})=\xi_{n+1}f\big(y_n-\varphi(y_n)\big)=\xi_{n+1}\psi(y_n),
\]
where $\psi(y_n)=f\big(y_n-\varphi(y_n)\big)$, and $\xi_n$ is chosen randomly from $\{\lambda_i:i\in I_N\}$ according to the probability measure $\pi$ on $I_N$.

Suppose there exists an invariant support set (related with the $\pi$-balanced measure for this Markov chain) for the stochastic process $y_{n+1}=\xi_{n+1}\psi(y_n)$, say an compact interval $X_Y$ located within $\R_+$. Define $g_i(y):=\lambda_i\psi(y)$ for all $i\in I_N$. Then we have an iterated random function system $(X_Y,G,\pi)$, where $G=\{g_i:i\in I_N\}$ with $g_i:X_Y\to X_Y$ for all $i\in I_N$, and $\pi$ is a probability measure on $I_N$.

We also have the dynamics of the optimal states of the capital in this economy
\[
k_{n+1}=y_n-\varphi(y_n)=\xi_nf(k_n)-\varphi\big(\xi_nf(k_n)\big)=\phi(k_n,\xi_n),
\]
where $\phi(k_n,\xi_n)=\xi_nf(k_n)-\varphi\big(\xi_nf(k_n)\big)$. Define $h_i(k):=\phi(k,\lambda_i)$ for all $i\in I_N$. Let the invariant support interval for $k_{n+1}=\phi(k_n,\xi_n)$ be $X_K\times\{\lambda_i:i\in I_N\}$. We have an iterated function system $(X_K,H)$, where $H=\{h_i:i\in I_N\}$ with $h_i:X_K\to X_k$ for all $i\in I_N$.

Consider a simple case, where $f(k)=\sqrt[3]{k}$, and $u(c)=\log c$. Suppose the shock $\xi_n$ is an i.i.d. Bernoulli process, such that $\xi_n=\lambda_a$ with probability $\pi(a)=q$ and $\xi_n=\lambda_b$ with probability $\pi(b)=1-q$, here $I_N=\{a,b\}$. We assume $1/\lambda_b>\lambda_a>1>\lambda_b>0$, so the index $a$ stands for a positive shock, and the $b$ stands for a negative shock, but the negative shock will not make the economy vanish, and at the same time the positive shock will not make it too expansive.

Then at the steady states of its optimal growth, we have $c_n=\varphi(y_n)=(1-\rho/3)y_n$, so $k_{n+1}=y_n-c_n=\rho y_n/3$. We know $y_n=\xi_n f(k_n)=\xi_n\sqrt[3]{k_n}$, so $k_{n+1}=\xi_n\rho\sqrt[3]{k_n}/3$. Define $\kappa_n=\log k_n$, then we have
\[
\kappa_{n+1}=\kappa_n/3+(\log\xi_n+\log\rho-\log 3).
\]
So when $\xi_n$ takes the value $\lambda_a$ or $\lambda_b$, we have two affine functions, namely $l_a(\kappa)=\kappa/3+(\log\lambda_a+\log\rho-\log3)$ and $l_b(\kappa)=\kappa/3+(\log\lambda_b+\log\rho-\log3)$. 

Suppose the invariant support interval of $l_a(\kappa)$ and $l_b(\kappa)$ is $[\alpha,\beta]$, then we have $\beta/3+(\log\lambda_a+\log\rho-\log3)=\beta$, and $\alpha/3+(\log\lambda_b+\log\rho-\log3)=\alpha$. After some calculations, we obtain $\alpha=3(\log\lambda_b+\log\rho-\log3)/2$, and $\beta=3(\log\lambda_a+\log\rho-\log3)/2$. So there exists an iterated function system $([\alpha,\beta],L)$, where $L=\{l_a,l_b\}$ with $l_i:[\alpha,\beta]\to[\alpha,\beta]$ for $i=a,b$. 

The difference between these two $l(\kappa)$-intercepts are $\log\lambda_a-\log\lambda_b=\log(\lambda_a/\lambda_b)$, and $\beta-\alpha=3(\log\lambda_a-\log\lambda_b)/2=3\log(\lambda_a/\lambda_b)/2$. So after some transformations, the function system $l_a(\kappa)$ and $l_b(\kappa)$ is equivalent with the pair of functions $q_a(\tau)=\tau/3$ and $q_b(\tau)=\tau/3+2/3$, where $\tau=\kappa/(\beta-\alpha)+d$, where $d$ is a constant. So the iterated function system $([\alpha,\beta],L)$ is equivalent with the iterated function system $([0,1],Q)$, where $Q=\{q_a,q_b\}$ with $q_i:[0,1]\to[0,1]$ for $i=a,b$. From Example \ref{exa3.1}, we know the attractor of $([0,1],Q)$ is the Cantor ternary set, so the attractor of $([\alpha,\beta],L)$ is also a Cantor set, and hence the dynamics of the steady states of $\kappa_t$ and also $k_t=\exp(\kappa_t)$ is in a chaotic situation. 

\section{Summary}
In this work, we study the random utility function and the optimal stochastic growth based on the formal developed theory of iterated function systems. Since the real economic world is very complex, we have to refine again and again our model settings and theoretical hypothesis to have a considerably thinkable framework. However, I wish these frameworks were able to capture some basic natures of the economic phenomena.

\end{document}